\documentclass[a4paper,12pt]{article}
\usepackage{amsthm,amsmath,amssymb}
\usepackage{enumitem}

\usepackage{url}
\usepackage{tikz}
\theoremstyle{plain}
\newtheorem{lemma}{Lemma}[section]
\newtheorem{theorem}[lemma]{Theorem}
\newtheorem{corollary}[lemma]{Corollary}
\newtheorem{proposition}[lemma]{Proposition}
\theoremstyle{definition}
\newtheorem{definition}[lemma]{Definition}
\newtheorem{remark}[lemma]{Remark}
\newtheorem{example}[lemma]{Example}
\newcommand{\F}{\mathbb F}
\newcommand{\wt}{\operatorname{wt}}
\newcommand{\Enc}{\operatorname{Enc}}
\newcommand{\Imf}{\operatorname{Im}}
\newcommand{\rank}{\operatorname{rank}}
\newcommand{\pr}{\operatorname{pr}}
\newcommand{\Nlin}{N_q^{\mathrm{lin}}}

\begin{document}

\title{On systematicity of linear function-correcting codes}
\author{Duy Ho \\
Department of Mathematical Sciences, \\
UAE University, PO Box 15551, Al Ain, UAE\\ 
Email: duyho@uaeu.ac.ae }
\date{}
\maketitle

\begin{abstract}
The standard formulation of function-correcting codes uses systematic encodings. 
We study the redundancy cost of this constraint
when both the prescribed function and the encoding are linear.
We introduce the function-separation distance and show that several classical unequal error protection parameters are special cases.
For linear functions and encodings, this distance is the first relative generalized Hamming weight of the code relative to the encoded kernel. 
We formulate free and systematic linear separation problems. 
We prove that the optimal free redundancy depends only on the rank of the function, and determine the optimal systematic redundancy for prescribed separations $d\leq3$. 
\end{abstract}

\textbf{Keywords}: 
Function-correcting codes, systematic encodings,
function-separation distance, unequal error protection, optimal redundancy, systematicity gap.

\textbf{Mathematics Subject Classification}: 94B05, 94B60, 94B65

\section{Introduction}
\label{sec:intro}
In 2023, Lenz, Bitar, Wachter-Zeh, and Yaakobi introduced
\emph{function-correcting codes} (FCCs), a coding framework designed
to protect the value of a prescribed function of a transmitted message against errors \cite{lenz2023}. 
In contrast to classical error-correcting codes, which impose a minimum-distance requirement between every pair of distinct codewords in order to recover the entire message,  FCCs impose such a requirement only on certain pairs of codewords. 
This relaxation is well motivated since there are scenarios  where the receiver is interested only in specific parts of the
transmitted message, rather than in recovering the entire message.

Since then, the theory of FCCs has developed along three main
directions. 
First, the framework has been extended to different
metrics and channel models, including the symbol-pair and
$b$-symbol metrics, homogeneous and Lee metrics over finite rings,
and the Rosenbloom-Tsfasman metric for parallel channels
\cite{hareesh2026,liu2026homogeneous,liu2026rt,
singh2025bsymbol,verma2026bsymbol,verma2025lee,
verma2026chainring,xia2024}.
Second, considerable attention has been devoted to
optimal redundancy and to particular classes of functions. Zhang
et al. improved bounds and constructions for the Hamming-weight and
Hamming-weight-distribution functions \cite{ge2025}, Ly and
Soljanin studied redundancy over general finite fields
\cite{ly2025}, and Premlal and Rajan obtained new bounds and investigated FCCs for linear functions \cite{premlal2025}. 
Third, several structural extensions of the FCC framework have been introduced. 
FCCs with data protection jointly protect the message
and a prescribed function and also include an investigation of
linear FCCs \cite{rajput2026}.
Function-correcting partition codes
formulate the problem directly in terms of partitions of the message space and allow several functions to be protected simultaneously \cite{rajputpartition2026}, while generalized function-correcting partition codes allow different partitions to receive different levels of protection
\cite{rajputgeneralized2026}.

The standard FCC formulation introduced in \cite{lenz2023} uses systematic encodings. This means each codeword has the form $(u,p(u))$, so that the original message coordinates are retained
and redundancy is appended to them. 
Although this requirement is
natural in applications, it imposes an additional structural constraint on the encoding and may increase the required redundancy.
In this paper, we study the cost of this constraint in
the linear setting. 
For a prescribed linear function $f$ and a
prescribed separation $d$, we compare arbitrary linear encodings with
linear encodings that are systematic relative to a distinguished basis of the
message space. We call the difference between the corresponding
optimal redundancies the \emph{systematicity gap}, and investigate
when this gap vanishes and how large it can be.

Let $V$ denote the message space. To compare these settings within a
common framework, for a function $f\colon V\to\Imf(f)$ and an encoding
$\Enc$, we consider the \emph{function-separation distance}
$\delta_{\Enc}(f)$, defined as the minimum Hamming distance between
encoded messages having distinct function values.
Given a prescribed separation $d$, the basic
requirement is therefore
\[
\delta_{\Enc}(f)\geq d.
\]
Depending on the linearity and systematicity conditions imposed on $f$ and $\Enc$, the general separation problem gives rise to the specializations represented in Figure~\ref{fig:problem-hierarchy}.
These problems are stated explicitly below.

\begin{enumerate}

\item[$\mathrm{I}$]
\emph{The separation problem.}
Given a function $f\colon V\to\Imf(f)$ and a prescribed separation
$d$, determine the least redundancy of an encoding $\Enc$ satisfying
$ \delta_{\Enc}(f)\geq d. $ 

\item[$\mathrm{II}$]
\emph{The systematic separation problem.}
Given a function $f\colon V\to\Imf(f)$ and a prescribed separation
$d$, determine the least redundancy of an encoding $\Enc$ that is
systematic relative to a distinguished basis $\mathcal E$ and satisfies  $\delta_{\Enc}(f)\geq d.$

For $d=2t+1$, this is the usual FCC problem, whose optimal redundancy is denoted by $r_f(k,t)$ in the literature.

\item[$\mathrm{III}$]
\emph{The free linear separation problem.}
Given a   linear function
$f\colon V\to\F_q^\ell$ and a prescribed separation $d$, determine
the least redundancy of a linear encoding $\Enc$ satisfying $
\delta_{\Enc}(f)\geq d. $

We denote the corresponding optimal
redundancy by $r_f^{\mathrm{free,lin}}(d)$.

\item[$\mathrm{II}\cap\mathrm{III}$]
\emph{The systematic linear separation problem.}
Given a   linear function
$f\colon V\to\F_q^\ell$ and a prescribed separation $d$, determine
the least redundancy of a linear encoding $\Enc$ that is systematic
relative to $\mathcal E$ and satisfies
$ \delta_{\Enc}(f)\geq d. $

We denote the corresponding optimal
redundancy by $r_{f,\mathcal E}^{\mathrm{sys,lin}}(d)$.

\end{enumerate}
Here, Problem~$\mathrm{II}\cap\mathrm{III}$ denotes the common
specialization of Problems~$\mathrm{II}$ and~$\mathrm{III}$.
Equivalently, it is the linear separation problem with the
additional requirement that the encoding be systematic.

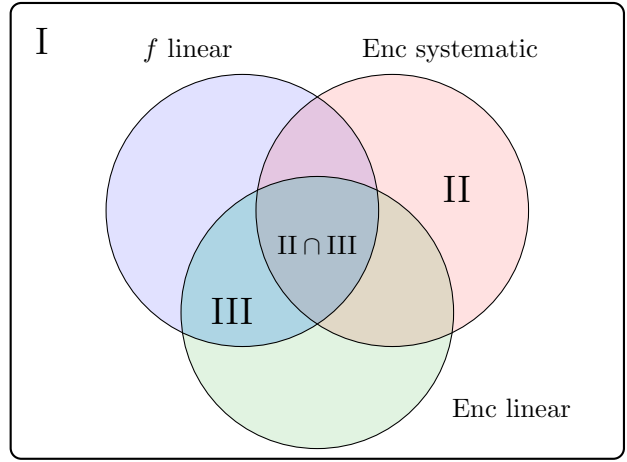
\begin{figure}[!h]
\centering
\begin{tikzpicture}[scale=0.9,
  every node/.style={font=\small}]
\draw[rounded corners, thick] (-4.5,-3.25) rectangle (4.5,3.45);
\node[anchor=north west, font=\large\bfseries] at (-4.30,3.25)
  {$\mathrm{I}$};

\begin{scope}[fill opacity=0.12]
\fill[blue] (-1.1,0.4) circle (2.0);              
\fill[red] (1.1,0.4) circle (2.0);               
\fill[green!60!black] (0,-1.1) circle (2.0);     
\end{scope}

\begin{scope}
\clip (-1.1,0.4) circle (2.0);
\fill[cyan!60, opacity=0.18] (0,-1.1) circle (2.0);
\end{scope}

\draw (-1.1,0.4) circle (2.0);
\draw (1.1,0.4) circle (2.0);
\draw (0,-1.1) circle (2.0);

\node[font=\footnotesize] at (-1.90,2.75) {$f$ linear};
\node[font=\footnotesize] at (1.95,2.75) {$\Enc$ systematic};
\node[font=\footnotesize] at (2.85,-2.50) {$\Enc$ linear};

\node[font=\large\bfseries] at (2.05,0.72) {$\mathrm{II}$};
\node[font=\large\bfseries] at (-1.25,-1.08) {$\mathrm{III}$};
\node[font=\footnotesize\bfseries] at (0,-0.10)
  {$\mathrm{II}\cap\mathrm{III}$};
\end{tikzpicture}
\caption{Specializations of the separation problem under linearity and systematicity constraints.}
\label{fig:problem-hierarchy}
\end{figure}

With the terminology above, in the present paper we address Problem~$\mathrm{III}$ and its systematic subclass Problem~$\mathrm{II}\cap\mathrm{III}$. 
We determine the optimal redundancy in Problem~$\mathrm{III}$
exactly for every prescribed separation $d$.
 For Problem~$\mathrm{II}\cap\mathrm{III}$, we establish general structural results and determine the optimal systematic linear redundancy for all prescribed separations $d\leq3$. At separation
$d=3$, our results characterize exactly when the systematicity gap
is zero and show that it can be arbitrarily large.

The relation between FCCs and classical unequal error protection (UEP) codes was noted in \cite[Sec.~I-B]{lenz2023}, see also
\cite{liu2026homogeneous, rajput2026}.
The function-separation distance makes this relation precise. 
We show that the standard UEP parameters are function-separation
distances of the corresponding coordinate and projection maps and, for
linear functions and linear encodings, identify this distance with a
relative generalized Hamming weight.

The content of the paper is organized as follows. 
In Section~\ref{sec:prelim}, we recall the basic definitions of
encodings, function-correcting codes, function-separation distance,
and the UEP parameters. 
In Section~\ref{sec:fcc-uep}, we establish the correspondence between function-separation distance and classical UEP parameters, and derive the relative-distance and separation-vector representations for linear encodings and linear functions. 
In Section~\ref{sec:systematicity}, we formulate
the free and systematic linear separation problems, determine the
optimal free linear redundancy, and study the effect of systematicity.
In Section~\ref{sec:distance-three}, we determine the optimal
systematic linear redundancy for prescribed separations $d\leq3$.

\section{Preliminaries}\label{sec:prelim}

Let $\F_q$ be a finite field of $q$ elements. 
For $x=(x_1,\dots,x_n)\in\F_q^n$, the \emph{Hamming weight} $\wt(x)$ is the number of nonzero coordinates $x_i$, and the \emph{Hamming distance} between $x,y\in\F_q^n$ is $d_H(x,y):=\wt(x-y)$.
All distances used in this paper are Hamming distances.

The message space is a $k$-dimensional $\F_q$-vector space $V$.
We fix a distinguished basis
\[
\mathcal E=(e_1,\dots,e_k)
\]
with associated coordinate map
$
[\,\cdot\,]_{\mathcal E}\colon V \rightarrow \F_q^k,
$
so that
$v=\sum_{i=1}^k\bigl([v]_{\mathcal E}\bigr)_i\,e_i$ for every $v\in V$.
We view each message $v\in V$ as the $k$-tuple
$x=[v]_{\mathcal E}\in\F_q^k$.

For a linear function $f\colon V\to W$, where $W$ is an
$\F_q$-vector space, we define its rank by
\[
\rank(f):=\dim_{\F_q}\Imf(f)
         =k-\dim_{\F_q}\ker f.
\]
In particular, if $f\colon V\to\F_q^\ell$ is surjective, then
$\rank(f)=\ell$.

Throughout the paper, $t$ and $d$ denote integers with $t\geq0$ and $d\geq1$. We will adopt the convention $\min\emptyset=+\infty$.

\subsection{Encodings, FCCs, and function-separation distance}\label{sec:common}

In many standard textbooks (for example \cite{ding2018,huffman2003,lingxing2004,macwilliams1977}), an encoding is introduced through its generator matrix, and thereby   linearity is assumed as part of the definition.
Because the redundancy comparisons in this paper also
require nonlinear encodings, we adopt the more general convention of Dunning and Robbins~\cite{dunning1978}, where an encoding is an arbitrary bijection from the message space $V$ onto its code and linearity is a special case.

For a basis $\mathcal B=(b_1,\dots,b_k)$ of $V$, let
\[
[\,\cdot\,]_{\mathcal B}\colon V\longrightarrow\F_q^k
\]
be the associated coordinate map, so that
$v=\sum_{i=1}^k\bigl([v]_{\mathcal B}\bigr)_i\,b_i$ for every $v\in V$.

\begin{definition} \label{def:encodings}
An \emph{encoding of length $n$} is a bijective map
$\Enc\colon V\to C$ onto a subset $C\subseteq\F_q^n$, called its
\emph{$(n,q^k)_q$ code}. Bijectivity forces $|C|=q^k$, hence $q^k\le q^n$ and $n\ge k$, and the integer
\[
r:=n-k\ge0
\]
is called the \emph{redundancy} of $\Enc$.

An encoding is a \emph{linear encoding} if it is linear as a map from $V$ into the ambient vector space $\F_q^n$.
In this case, the code $C$ is  a linear code of dimension $k$ and we call $C$ an $[n,k]_q$ code. 
The \emph{generator matrix} of $\Enc$ relative to $\mathcal B$ is the
$k\times n$ matrix
\[
G_{\mathcal B}(\Enc)
:=
\begin{pmatrix}
\Enc(b_1)\\
\vdots\\
\Enc(b_k)
\end{pmatrix},
\]
so that
$\Enc(v)=[v]_{\mathcal B}\,G_{\mathcal B}(\Enc)$
for $v \in V$. This is the one-to-one correspondence  between linear encodings and   generator matrices of a linear code in~\cite{dunning1978}, once the basis $\mathcal B$ is chosen.

An encoding $\Enc\colon V\to C\subseteq\F_q^{k+r}$ is \emph{systematic relative
to (the distinguished basis) $\mathcal E$} if 
\[
\Enc(v)=\bigl([v]_{\mathcal E},p(v)\bigr)
\]
for some map $p\colon V\to\F_q^r$. 
If $\Enc$ is linear, this is equivalent to
$G_{\mathcal E}(\Enc)=[\,I_k\mid P\,]$ for some
$P\in\F_q^{k\times r}$.
Since $\mathcal E$ is fixed, we
usually suppress the phrase ``relative to $\mathcal E$.''
\end{definition}

\begin{definition}
\label{def:fcc}
Let $f\colon V\to\Imf(f)$ be a function. A systematic encoding
$\Enc\colon V\to C\subseteq\F_q^{k+r}$ relative to $\mathcal E$ is an
\emph{$(f,t)$-function-correcting
code} if
\[
f(u)\neq f(v)
\quad\Longrightarrow\quad
 d_H\bigl(\Enc(u),\Enc(v)\bigr)\geq 2t+1
\]
for all $u,v\in V$.
\end{definition}

\begin{definition}
\label{def:optimal-redundancy}
For $f\colon V\to\Imf(f)$ and $t\geq0$, the \emph{optimal
redundancy} $r_f(k,t)$ is the smallest $r$ such that there exists an $(f,t)$-function-correcting code
$\Enc\colon V\to C\subseteq\F_q^{k+r}$.
\end{definition}

Definition~\ref{def:fcc} was  originally introduced in
\cite{lenz2023} over $\F_2$ and subsequently extended to $\F_q$ in \cite{premlal2025,rajput2026}. 
Because systematicity is defined relative to the fixed basis $\mathcal E$, the same dependence is understood in the notation $r_f(k,t)$.

We now introduce the function-separation distance for a general (not necessarily linear) encoding. This is the main parameter we study in this paper.  
\begin{definition}
\label{def:function-separation}
Let $\Enc\colon V\to C\subseteq\F_q^n$ be an encoding and let
$f\colon V\to\Imf(f)$ be a function. The
\emph{function-separation distance of $\Enc$ with respect to $f$} is
\[
\delta_{\Enc}(f)
:=
\min_{\substack{u,v\in V\\ f(u)\neq f(v)}}
d_H\bigl(\Enc(u),\Enc(v)\bigr).
\]
For $d\geq1$, we say that $\Enc$ \emph{achieves separation $d$ with
respect to $f$} if
\[
\delta_{\Enc}(f)\geq d.
\]
We call $d$ the \emph{prescribed separation}.
\end{definition}

\begin{proposition}
\label{prop:separation-criterion}
Let $\Enc\colon V\to C\subseteq\F_q^n$ be an encoding, $f\colon V\to\Imf(f)$ a
function, and $t\ge 0$. Then the following are equivalent:
\begin{enumerate} 
\item there exists a map $g\colon\F_q^n\to\Imf(f)$ such that
$g(y)=f(v)$ for every $v\in V$ and every $y\in\F_q^n$ with
$d_H\bigl(y,\Enc(v)\bigr)\le t$;
\item $\delta_{\Enc}(f)\ge 2t+1$.
\end{enumerate}
\end{proposition}
\begin{proof}
For $x\in\F_q^n$ write $B_t(x)=\{y\in\F_q^n:d_H(y,x)\le t\}$. 
A map $g$ as in~(1) exists if and only if 
$B_t(\Enc(u))\cap B_t(\Enc(v))=\emptyset$ whenever $f(u) \ne f(v)$. 
By the triangle inequality, this is equivalent to
$d_H\bigl(\Enc(u),\Enc(v)\bigr)\ge 2t+1$ whenever $f(u) \ne f(v)$, which is $\delta_{\Enc}(f)\ge 2t+1$.
\end{proof}
Thus achieving separation $2t+1$ with respect to $f$ is equivalent
to recovering the function value in the presence of at most $t$
errors.

 An important feature of the function-separation distance is that it depends on $f$ only through the pairs of messages it separates. This follows directly from the definition of $\delta$, and we record this feature in the following. 
\begin{lemma}\label{prop:fiber-invariance}
Let $f\colon V\to A$ and $g\colon V\to B$ be functions with
\[
f(u)=f(v)\quad\Longleftrightarrow\quad g(u)=g(v)
\qquad\text{for all }u,v\in V.
\]
Then $\delta_{\Enc}(f)=\delta_{\Enc}(g)$ for every encoding
$\Enc$.
\end{lemma} 

As a consequence,  when $f$ is linear, the preimages $f^{-1}(y), y \in \Imf(f),$  are the cosets of $\ker f$, so the function-separation distance $\delta$ is determined by the kernel alone.
\begin{corollary}\label{cor:kernel-invariance}
Let $f\colon V\to A$ and $g\colon V\to B$ be linear functions with
$\ker f=\ker g$. Then $\delta_{\Enc}(f)=\delta_{\Enc}(g)$ for every
encoding $\Enc$.
\end{corollary} 

\subsection{UEP parameters}  \label{sec:uepdef}
In~\cite{masnick1967}, Masnick and Wolf assign to each output coordinate of a linear code an error protection level $f_i$.
 Instead of using $f_i$ directly, we record the parameter $a_i(C)$ as in Definition \ref{def:output-separation} below. We will explain the reason for this choice in Remark \ref{rem:masnick-wolf-level}.
\begin{definition} \label{def:output-separation}
For a linear code $C\subseteq\F_q^n$ and $i\in\{1,\dots,n\}$, the
\emph{$i$th output separation} of $C$ is
\[
a_i(C):=
\min\{\wt(c):c\in C,\ c_i\neq0\}.
\]
\end{definition}

Dunning and Robbins~\cite{dunning1978} introduced the separation vector to guarantee protection for the input message in specified positions. We recall its definition. 
\begin{definition}
\label{def:input-separation}
Let $\Enc\colon V\to C\subseteq\F_q^n$ be an encoding and $\mathcal B=(b_1,\dots,b_k)$
a basis of $V$. For $1\le i\le k$, the \emph{$i$th input separation} of $\Enc$
relative to $\mathcal B$ is
\[
s_i^{\mathcal B}(\Enc)
:=\min\bigl\{\,d_H\bigl(\Enc(u),\Enc(v)\bigr)\;:\;u,v\in V,\
\bigl([u]_{\mathcal B}\bigr)_i\neq\bigl([v]_{\mathcal B}\bigr)_i\,\bigr\},
\]
and $s^{\mathcal B}(\Enc)=(s_1^{\mathcal B}(\Enc),\dots,s_k^{\mathcal B}(\Enc))$
is the \emph{separation vector} of $\Enc$. 
\end{definition}
When $\Enc$ is linear with generator matrix $G=G_{\mathcal B}(\Enc)$, each entry of the separation vector takes the form
\[
s_i(G):=s_i^{\mathcal B}(\Enc)=\min\{\wt(aG):a\in\F_q^k,\ a_i\neq0\}.
\]
This is the notion of separation vector used in  van Gils \cite{vangils1983}. 

In~\cite{linlin1988}, Lin and Lin  partition the message space
into a product $A_1\times\cdots\times A_m$ of binary blocks
$A_i=\F_2^{k_i}$ and measure the protection of the $i$th block by a separation vector with $m$ components, which we generalize to $\F_q$ as follows.  
\begin{definition}
\label{def:component-separation}
Let $\mathcal V=(V_1,\dots,V_m)$ be a direct-sum decomposition of $V$, so that $ V=V_1\oplus\cdots\oplus V_m,$ and let $\rho_i\colon V\to V_i$ be the $i$th component projection.
For an encoding $\Enc\colon V\to C\subseteq\F_q^n$,
   the \emph{$i$th component separation} is
\[
s_i^{\mathcal V}(\Enc)
:=
\min_{\substack{u,v\in V\\
                 \rho_i(u)\neq\rho_i(v)}}
 d_H\bigl(\Enc(u),\Enc(v)\bigr),
\qquad 1\leq i\leq m.
\]
The vector
\[
\mathbf s^{\mathcal V}(\Enc)
:=
\bigl(s_1^{\mathcal V}(\Enc),\dots,
      s_m^{\mathcal V}(\Enc)\bigr)
\]
is the \emph{component  separation vector} of $\Enc$
relative to $\mathcal V$.
\end{definition}

\section{A correspondence between function-separation distances and UEP parameters}
\label{sec:fcc-uep}

In this section, we first motivate the function-separation viewpoint by showing that the three classical UEP parameters $a_i, s_i,$ and $s_i^\mathcal{V}$ in Subsection~\ref{sec:uepdef} are
function-separation distances of the corresponding coordinate and projection maps. This is the content of Subsection \ref{sec:dictionary}. 

In Subsection~\ref{sec:subspace} we generalize these observations, for linear encodings,  from coordinate and projection maps to arbitrary linear functions. Conversely, for a linear encoding and a linear function $f$, the function-separation distance is a UEP parameter under a suitable basis change.

\subsection{UEP parameters as function-separation distances}
\label{sec:dictionary}

For $i\in\{1,\dots,n\}$, let $\pr_i\colon\F_q^n\to\F_q$ denote the projection onto the $i$th
coordinate. 
The map $\pr_i\circ\Enc$ is the $i$th transmitted symbol viewed as a function of
the source message.
 
\begin{proposition}\label{thm:output-uep}
Let $\Enc\colon V\to C\subseteq\F_q^n$ be a linear encoding. Then for $1 \le i \le n$, 
\[
\delta_{\Enc}(\pr_i\circ\Enc)=a_i(C).
\]
\end{proposition}

\begin{proof} Let $u,v \in V$,  $w:=u-v$ and let $c:=\Enc(w)$. 
Since $\Enc$ is linear,  
\[ 
(\pr_i\circ\Enc)(u)\neq(\pr_i\circ\Enc)(v)
 \iff (\Enc(w))_i\neq0 \iff c_i \ne 0. 
\]
Also, $d_H\bigl(\Enc(u),\Enc(v)\bigr)=\wt\bigl(\Enc(w)\bigr)=\wt\bigl(c\bigr)$.
Since $(u,v)$ ranges over $V \times V$, the difference $w$ ranges over $V$, and taking minimum gives the claim.
\end{proof}

\begin{remark}\label{rem:masnick-wolf-level}
We note that Masnick and Wolf~\cite{masnick1967} assign the $i$th output coordinate an error protection level $f_i$, which is the number of errors against which it
is guaranteed correct. 
In our notation, this means $f_i=\lfloor(a_i(C)-1)/2\rfloor$.
The map $a_i(C)\mapsto f_i$ is two-to-one. 
Consequently, the level $f_i$ determines the function-separation distance $\delta_{\Enc}(\pr_i\circ\Enc)$ only up to the two values $2f_i+1$ and
$2f_i+2$, whereas $a_i(C)$ determines it exactly.
\end{remark}

 For a basis $\mathcal B=(b_1,\dots,b_k)$ of $V$, $i\in\{1,\dots,k\}$, and
$\emptyset\neq J\subseteq\{1,\dots,k\}$, let
$\pi_i^{\mathcal B}(v):=([v]_{\mathcal B})_i$ be the $i$th coordinate of
$v$ in that basis, and let $\pi_J^{\mathcal B}(v):=([v]_{\mathcal B})|_J$
record the coordinates indexed by $J$.

\begin{proposition}\label{thm:input-uep}
Let $\Enc\colon V\to C\subseteq\F_q^n$ be an encoding and $\mathcal B$ a basis of
$V$. Then for $1\le i\le k$,
\[
\delta_{\Enc}(\pi_i^{\mathcal B})=s_i^{\mathcal B}(\Enc).
\]
In particular, if $\Enc$ is linear, then
$\delta_{\Enc}(\pi_i^{\mathcal B})=s_i\bigl(G_{\mathcal B}(\Enc)\bigr)$.
\end{proposition}

\begin{proof}
The identity $\delta_{\Enc}(\pi_i^{\mathcal B})=s_i^{\mathcal B}(\Enc)$ follows from
Definitions~\ref{def:function-separation}
and~\ref{def:input-separation}. The linear case follows from the
matrix form of the separation vector in
Subsection~\ref{sec:uepdef}.
\end{proof}

\begin{corollary}\label{cor:block-projection}
Let $\Enc\colon V\to C\subseteq\F_q^n$ be an encoding, let $\mathcal B$ be a
basis of $V$, and let $\emptyset\neq J\subseteq\{1,\dots,k\}$. Then
\[
\delta_{\Enc}(\pi_J^{\mathcal B})
=
\min_{i\in J}s_i^{\mathcal B}(\Enc).
\]
In particular, if $\Enc$ is linear, then
$\delta_{\Enc}(\pi_J^{\mathcal B})
=\min_{i\in J}s_i\bigl(G_{\mathcal B}(\Enc)\bigr)$.
\end{corollary}

\begin{proposition}\label{thm:component-uep}
Let $\mathcal V=(V_1,\dots,V_m)$ be a direct-sum decomposition of
$V$, let $\rho_i\colon V\to V_i$ be the component projections, and
let $\Enc\colon V\to C\subseteq\F_q^n$ be an encoding. Then
for $1\leq i\leq m$,
\[
\delta_{\Enc}(\rho_i)=s_i^{\mathcal V}(\Enc).
\]
\end{proposition}

\begin{proof}
The two sides are the same minimum by
Definitions~\ref{def:function-separation} and
\ref{def:component-separation}.
\end{proof}

\subsection{Function-separation distance and the separation-vector representation for linear encodings}
\label{sec:subspace}
In this subsection we restrict to linear encodings. We first generalize the identifications above from coordinate
and projection maps to arbitrary linear functions.
For nested linear codes $D\subsetneq C\subseteq\F_q^n$, we define
\[
d(C,D):=\min\{\wt(c):c\in C\setminus D\}.
\]
This is the first relative generalized Hamming weight $M_1(C,D)$ of~\cite{luo2005}, in the minimum-weight form
of~\cite[Lemma~1]{liu2008} (see also \cite{geil2014,zhuang2013}), and coincides with the coset-code distance $d(C/D)$ of~\cite[Definition~17]{rajput2026}.

\begin{theorem}
\label{thm:kernel-formula}
Let $\Enc\colon V\to C \subseteq \F_q^n$ be a linear encoding. 
Let $f\colon V\to\F_q^\ell$ be a nonzero linear function.
Then
\[
\delta_{\Enc}(f) = d(C,\Enc(\ker f)).
\]
\end{theorem}

\begin{proof} Let $U:=\ker f$ and $D:=\Enc(U)$.  
Let $u,v\in V$, and  $w:=u-v$. Since $f$ is linear, we have
\[
f(u)\neq f(v)
\quad\Longleftrightarrow\quad
f(w)\neq0
\quad\Longleftrightarrow\quad
w\notin U.
\]
Since $\Enc$ is linear, we have
$
d_H\bigl(\Enc(u),\Enc(v)\bigr)
=
\wt\bigl(\Enc(w)\bigr).
$
As $(u,v)$ ranges over $V \times V$, the difference $w$ ranges over $V$.
Consequently,
\[
\delta_{\Enc}(f)
=
\min_{w\in V\setminus U}\wt\bigl(\Enc(w)\bigr).
\]
The map $\Enc$ is a bijection from $V$ onto $C$. Hence,
$
\Enc(V\setminus U)=C\setminus D,
$
and so 
\[
\delta_{\Enc}(f)
=
\min\{\wt(c):c\in C\setminus D\}
=
d(C,D).
\]
This proves the theorem.
\end{proof}

\begin{remark}\label{rem:kernel-specializations}
For linear encodings, the identifications of
Subsection~\ref{sec:dictionary} are instances of
Theorem~\ref{thm:kernel-formula}.  
Taking $f=\pr_i\circ\Enc$, with
$\Enc(\ker f)=\{c\in C:c_i=0\}$, recovers
Proposition~\ref{thm:output-uep}.
Taking $f=\pi_i^{\mathcal B}$ and
$f=\pi_J^{\mathcal B}$ recovers the linear cases of
Proposition~\ref{thm:input-uep} and
Corollary~\ref{cor:block-projection}.
Taking $f=\rho_i$ recovers the linear
case of Proposition~\ref{thm:component-uep}. 
For arbitrary encodings
however, Propositions~\ref{thm:input-uep}
and~\ref{thm:component-uep} are not covered by
Theorem~\ref{thm:kernel-formula}.
\end{remark}

Specializing Theorem~\ref{thm:kernel-formula} to systematic encodings, we obtain the following characterization.
\begin{corollary}
\label{cor:linear-fcc-criterion}
Let $\Enc\colon V\to C\subseteq\F_q^{k+r}$ be a linear encoding. Let $f\colon V\to \F_q^\ell$ be
a nonzero linear function. Assume that $\Enc$ is systematic
relative to $\mathcal E$. Then for every 
$t\geq0$, the encoding $\Enc$ is an $(f,t)$-function-correcting code
if and only if
\[
d(C,\Enc(\ker f))\geq 2t+1.
\]
\end{corollary}
Corollary~\ref{cor:linear-fcc-criterion} recovers the characterization of systematic linear FCCs for linear
functions in Rajput et al. \cite[Section VII.A]{rajput2026}. 
Their coset-code distance $d(C/D_f)$ coincides with the  distance
$d\bigl(C,\Enc(\ker f)\bigr)$.

By Corollary~\ref{cor:kernel-invariance}, replacing a linear
function by a surjective linear function with the same kernel does
not change its function-separation distance for any encoding.
Thus, after identifying $\Imf(f)$ with $\F_q^\ell$, where
$\ell=\rank(f)$, we can assume without loss
of generality that  $f\colon V\to\F_q^\ell$ is surjective.

We now show that with a suitable choice of basis for $V$, the function-separation distance $\delta_{\Enc}(f)$ can be expressed in terms of input separations of the generator matrix. 

\begin{theorem}\label{thm:adapted-basis}
Let $\Enc\colon V\to C\subseteq\F_q^n$ be a linear encoding, let
$f\colon V\to\F_q^\ell$ be a surjective linear function.
Let $U=\ker f$.
Let $\mathcal B=(b_1,\dots,b_k)$ be a basis of $V$ with
$U=\operatorname{span}\{b_{\ell+1},\dots,b_k\}$. Then
\[
\delta_{\Enc}(f)
=
\min_{1\leq i\leq \ell}s_i\bigl(G_{\mathcal B}(\Enc)\bigr).
\]

\end{theorem}

\begin{proof}
Once the basis $\mathcal{B}$ is fixed, for $v\in V$ we write $a=(a_1,\dots,a_k):=[v]_{\mathcal B}$, so that
$v=a_1b_1+\cdots+a_kb_k$.
By Definition~\ref{def:encodings},
$\Enc(v)=aG_{\mathcal B}(\Enc)$.
Since $b_{\ell+1},\dots,b_k$ form a
basis of $U$,
\[
v\in U
\quad\Longleftrightarrow\quad
a_1=\cdots=a_\ell=0.
\]
Hence, by Theorem~\ref{thm:kernel-formula},
\[
\delta_{\Enc}(f)
=
\min_{v\in V\setminus U}\wt\bigl(\Enc(v)\bigr)
=
\min_{\substack{a\in\F_q^k\\(a_1,\dots,a_\ell)\neq0}}
\wt\bigl(aG_{\mathcal B}(\Enc)\bigr).
\]
Since
$\{a\in\F_q^k:(a_1,\dots,a_\ell)\neq0\}
=\bigcup_{i=1}^\ell\{a\in\F_q^k:a_i\neq0\}$, Definition~\ref{def:input-separation} gives
\[
\min_{\substack{a\in\F_q^k\\(a_1,\dots,a_\ell)\neq0}}
\wt\bigl(aG_{\mathcal B}(\Enc)\bigr)
=
\min_{1\leq i\leq \ell}s_i\bigl(G_{\mathcal B}(\Enc)\bigr),
\]
which proves the theorem.
\end{proof}
We call   $\min_{1\leq i\leq \ell}s_i\bigl(G_{\mathcal B}(\Enc)\bigr)$   the \emph{separation-vector
representation} of $\delta_{\Enc}(f)$ with respect to the basis $\mathcal B$.

\begin{remark} 
\label{rem:basis-change-systematicity}
Let $M_{\mathcal B}\in\F_q^{k\times k}$ be the invertible matrix
whose $i$th row is $[b_i]_{\mathcal E}$. Then
\[
G_{\mathcal B}(\Enc)=M_{\mathcal B}\,G_{\mathcal E}(\Enc).
\]
So if $\Enc$ is systematic relative to $\mathcal E$, say
$G_{\mathcal E}(\Enc)=[\,I_k\mid P\,]$, then
\[
G_{\mathcal B}(\Enc)=[\,M_{\mathcal B}\mid M_{\mathcal B}P\,],
\]
which has systematic form only when $\mathcal B=\mathcal E$. 
Hence an encoding $\Enc$ that is systematic relative to $\mathcal E$ is not systematic relative
to any basis $\mathcal B\neq\mathcal E$. 
It is this incompatibility  that motivates us to introduce and compare free and systematic redundancies in the next section. 
\end{remark}

\section{Optimal redundancy and the systematicity gap}
\label{sec:systematicity}
As in Section \ref{sec:subspace}, we assume throughout this section without loss
of generality that  $f\colon V\to\F_q^\ell$ is surjective.
We now consider the following problem.

\begin{quote}
\emph{The free linear separation problem.}
Given a surjective linear function $f\colon V\to\F_q^\ell$ and a
prescribed separation $d$, determine the least redundancy of a
linear encoding $\Enc$ with $\delta_{\Enc}(f)\geq d$.
\end{quote}

The answer depends on whether the encoding is required to be systematic
relative to $\mathcal E$, and we introduce the two corresponding
optimal redundancies.

\begin{definition}
\label{def:free-systematic-redundancy}
Let $f\colon V\to\F_q^\ell$ be a linear function.

The \emph{optimal free linear redundancy} $r_f^{\mathrm{free,lin}}(d)$
is the smallest $r\geq0$ such that there exists a linear encoding
$\Enc\colon V\to C\subseteq\F_q^{k+r}$ with $\delta_{\Enc}(f)\geq d$.

The \emph{optimal systematic linear redundancy}
$r_{f,\mathcal E}^{\mathrm{sys,lin}}(d)$ is the smallest $r\geq0$
such that there exists a linear encoding
$\Enc\colon V\to C\subseteq\F_q^{k+r}$, systematic relative to $\mathcal E$,
with $\delta_{\Enc}(f)\geq d$.

We call the difference
\[
r_{f,\mathcal E}^{\mathrm{sys,lin}}(d)
-
r_f^{\mathrm{free,lin}}(d)
\]
the \emph{systematicity gap} of $f$ relative to
$\mathcal E$ at prescribed separation $d$.
Since $\mathcal E$ is fixed, we write $r_f^{\mathrm{sys,lin}}(d)$
for $r_{f,\mathcal E}^{\mathrm{sys,lin}}(d)$.
\end{definition}

\begin{remark}
\label{rem:optimal-redundancy-relation}
Let $f$ be a linear function and $d=2t+1$. 
The minimum in
Definition~\ref{def:optimal-redundancy} runs over all systematic encodings, not necessarily linear, so
\[
r_f(k,t)\leq r_f^{\mathrm{sys,lin}}(2t+1).
\]
For every $d\geq1$, it is also clear that
\[
r_f^{\mathrm{free,lin}}(d)\leq r_f^{\mathrm{sys,lin}}(d).
\] 
However, in general  we cannot compare $r_f(k,t)$ and $r_f^{\mathrm{free,lin}}(2t+1)$, since the former requires systematicity for encodings.
\end{remark}

For $\ell\geq1$ and $d\geq1$, let
\[
\Nlin(\ell,d)
:=
\min\{N:\text{there exists a linear }[N,\ell,\geq d]_q\text{ code}\}.
\]
Following van Gils~\cite[Definition~4]{vangils1983}, for a  vector
$\mathbf s=(s_1,\dots,s_m)\in\mathbb N^m$, let $n_q(\mathbf s)$
denote the shortest length of an $m$-dimensional linear code over
$\F_q$ admitting a generator matrix $G$ with
$s_i(G)\geq s_i$ for $1\leq i\leq m$. In particular,
$n_q(d,\dots,d)=\Nlin(\ell,d)$ for $\ell$-dimensional codes. We further record two properties of $n_q$ from  van Gils~\cite[Theorem~4, Corollary~9]{vangils1983}.
\begin{lemma}
\label{lem:vangils-bounds}
Let $\mathbf s\in\mathbb N^m$.
\begin{enumerate}
\item For any partition $\mathbf s=(\mathbf s_1\mid\mathbf s_2)$,
\[
n_q(\mathbf s_1\mid\mathbf s_2)\leq n_q(\mathbf s_1)+n_q(\mathbf s_2).
\]
\item If $\mathbf s$ is nonincreasing, then for $1\leq j<m$,
\[
n_q(s_1,\dots,s_m)\geq j+n_q(s_1,\dots,s_{m-j}).
\]
\end{enumerate}
\end{lemma}

\begin{theorem}
\label{thm:free-exact}
Let $f\colon V\to\F_q^\ell$ be a surjective linear function. For every $r\geq0$, the following are equivalent:
\begin{enumerate}
\item there exists a linear encoding $ \Enc\colon V\to C\subseteq\F_q^{k+r} $
with $\delta_{\Enc}(f)\geq d$;
\item there exists a linear $[r+\ell,\ell,\geq d]_q$ code.
\end{enumerate}
Consequently,
\[
r_f^{\mathrm{free,lin}}(d)
=
\Nlin(\ell,d)-\ell.
\]
\end{theorem}
\begin{proof}
By Theorem~\ref{thm:adapted-basis} and the correspondence between
linear encodings and generator matrices
(Definition~\ref{def:encodings}), an encoding as in~(1) exists if and
only if some $k\times(k+r)$ generator matrix $G$ satisfies
$s_i(G)\geq d$ for $1\leq i\leq\ell$. Since every injective encoding
has all input separations at least $1$, this holds if and only if
\[
n_q(\underbrace{d,\dots,d}_{\ell},\underbrace{1,\dots,1}_{k-\ell})
\leq k+r.
\]
On the other hand, by Lemma~\ref{lem:vangils-bounds},  we have
\[
n_q(d,\dots,d,1,\dots,1)
=
\Nlin(\ell,d)+(k-\ell).
\]
Hence an encoding as in~(1) exists if and only if
$r+\ell\geq\Nlin(\ell,d)$, which means a linear
$[r+\ell,\ell,\geq d]_q$ code exists.
\end{proof}
 
\begin{remark}
After choosing a basis $\mathcal{B}$ as in Theorem~\ref{thm:adapted-basis}, the free linear separation problem becomes the shortest-length problem of
van Gils~\cite[Definition~4]{vangils1983} for the separation vector
\[
(\underbrace{d,\dots,d}_{\ell},
 \underbrace{1,\dots,1}_{k-\ell}).
\]
In this sense, Theorem~\ref{thm:free-exact} expresses the optimal free linear redundancy through a classical UEP quantity and shows that it depends on
$f$ only through $\rank(f)$. The additional feature studied here is
the requirement of systematicity relative to the distinguished basis.
\end{remark}

We obtain a few corollaries from Theorem~\ref{thm:free-exact}. The first is a  lower bound for the systematic linear redundancy.

\begin{corollary}
\label{cor:systematic-lower}
Let $f\colon V\to\F_q^\ell$ be a surjective linear function. Then
\[
r_f^{\mathrm{sys,lin}}(d)
\geq
\Nlin(\ell,d)-\ell.
\]
\end{corollary}
For $\ell=1$, this lower bound is attained by a systematic linear encoding.
\begin{corollary} 
\label{cor:rank-one-coordinate}
Let $f\colon V\to\F_q$ be a surjective linear function. Then for every
$d\geq1$,
\[
r_f^{\mathrm{sys,lin}}(d)=d-1.
\] 
\end{corollary}

\begin{proof}
Since $\Nlin(1,d)=d,$
Corollary~\ref{cor:systematic-lower} gives the lower bound
$r_f^{\mathrm{sys,lin}}(d)\geq d-1$.
For the upper bound, define
\[
\Enc(v)
=
\bigl([v]_{\mathcal E},
\underbrace{f(v),\dots,f(v)}_{d-1\text{ copies}}\bigr).
\]
This encoding is linear and systematic relative to $\mathcal E$. If
$f(u)\neq f(v)$, then $w=u-v$ is nonzero and each of the last $d-1$
coordinates of $\Enc(w)$ is nonzero. 
Hence $ \wt\bigl(\Enc(w)\bigr)\geq 1+(d-1)=d, $
and so $\delta_{\Enc}(f)\geq d$.  
\end{proof}

By Remark~\ref{rem:optimal-redundancy-relation} and
Corollary~\ref{cor:rank-one-coordinate},
$r_f(k,t)\leq r_f^{\mathrm{sys,lin}}(2t+1)=2t$.
On the other hand, in \cite[Corollary 1]{lenz2023}, it is shown that $r_f(k,t)\geq 2t$ for any function $f$ with $| \Imf(f)| \ge2$. 
The statement in \cite{lenz2023} is over $\F_2$, but   the proof can be applied to any field $\F_q$, cf.
\cite{premlal2025}. We obtain the following.
\begin{corollary}
\label{cor:rank-one-fcc-optimal}
Let $f\colon V\to\F_q$ be a surjective linear function. Then
\[
r_f(k,t)=2t,
\]
and the optimal redundancy is attained by a systematic linear  encoding.
\end{corollary}

Combining Theorem~\ref{thm:free-exact} with the Singleton and
Griesmer bounds yields the following.
\begin{corollary} 
\label{cor:free-bounds}
Let $\Enc\colon V\to C\subseteq\F_q^n$ be a linear encoding, and let
$f\colon V\to\F_q^\ell$ be a surjective linear function. If
$\delta_{\Enc}(f)\geq d$, then the redundancy $r$ of $\Enc$
satisfies
\[
r\geq d-1
\qquad\text{and}\qquad
r+\ell
\geq
\sum_{i=0}^{\ell-1}
\left\lceil\frac{d}{q^i}\right\rceil.
\]
\end{corollary}

We have seen in Corollary~\ref{cor:kernel-invariance} that
functions with the same kernel have the same function-separation
distance under every encoding, and in
Theorem~\ref{thm:free-exact} that the free linear redundancy
$r_f^{\mathrm{free,lin}}(d)$ is the same for all surjective linear
functions of the same rank. 
We now consider when two linear
functions have the same systematic redundancy. 
For a linear function $f\colon V\to\F_q^\ell$, let
$M_f\in\F_q^{k\times\ell}$ denote the matrix with rows
$f(e_1),\dots,f(e_k)$, so that $f(v)=[v]_{\mathcal E}M_f$.
We say that two linear functions  $f,g: V \rightarrow \F_q^\ell$ are
\emph{monomially equivalent} if $M_g=AM_f$ for a monomial matrix
$A\in\F_q^{k\times k}$. Equivalently,    there is a linear automorphism $M\in\mathrm{GL}(V)$ whose matrix relative to
$\mathcal E$ is monomial such that $ g=f\circ M.$ 

The next theorem shows that monomially equivalent functions have the same systematic redundancy at every separation.

\begin{theorem}
\label{thm:monomial-invariance}
Let $f,g\colon V\to\F_q^\ell$ be monomially equivalent surjective
linear functions. Then  for every $d\geq1$,
\[
r_f^{\mathrm{sys,lin}}(d)=r_g^{\mathrm{sys,lin}}(d).
\]
\end{theorem}

\begin{proof}
Write $g=f\circ M$, where $M\in\mathrm{GL}(V)$ has monomial matrix
$A$ relative to $\mathcal E$, so that
$[M(v)]_{\mathcal E}=[v]_{\mathcal E}A$. Let
$\Enc(v)=\bigl([v]_{\mathcal E},p(v)\bigr)$ be a linear encoding
systematic relative to $\mathcal E$, with redundancy $r$ and
$\delta_{\Enc}(f)\geq d$. Define
\[
\widetilde\Enc(v)
:=
\bigl([v]_{\mathcal E},p(M(v))\bigr).
\]
Since $p\circ M$ is linear, $\widetilde\Enc$ is a linear encoding,
systematic relative to $\mathcal E$, with redundancy $r$.
For every $w\in V$, we have
\[
\begin{aligned}
\wt\bigl(\widetilde\Enc(w)\bigr)
&=
\wt\bigl([w]_{\mathcal E}\bigr)+\wt\bigl(p(M(w))\bigr)\\
&=
\wt\bigl([M(w)]_{\mathcal E}\bigr)+\wt\bigl(p(M(w))\bigr)
=
\wt\bigl(\Enc(M(w))\bigr).
\end{aligned}
\]
Moreover, 
$w\mapsto M(w)$ is a bijection of $\{w:g(w)\neq0\}$ onto
$\{w:f(w)\neq0\}$. Hence
\[
\delta_{\widetilde\Enc}(g)
=
\min_{g(w)\neq0}\wt\bigl(\widetilde\Enc(w)\bigr)
=
\min_{f(w')\neq0}\wt\bigl(\Enc(w')\bigr)
=
\delta_{\Enc}(f)
\geq d.
\]
This implies that 
$r_g^{\mathrm{sys,lin}}(d)\leq r_f^{\mathrm{sys,lin}}(d)$.  
The  same argument applied to $f=g\circ M^{-1}$ gives the
reverse inequality.
\end{proof}

By Theorem~\ref{thm:adapted-basis}, every surjective linear function $f$ admits a separation-vector representation of
$\delta_{\Enc}(f)$ with respect to a suitable basis $\mathcal{B}$, and
Remark~\ref{rem:basis-change-systematicity} showed that an encoding
systematic relative to $\mathcal E$ is  not
systematic relative to such a basis. 
In the next theorem, we consider functions $f$
admitting a separation-vector representation with respect to a reordering of $\mathcal E$ itself. 
These functions are simply the ones whose kernel is spanned by distinguished basis vectors. 
For such functions, we will show that  $r_f^{\mathrm{sys,lin}}(d) = r_f^{\mathrm{free,lin}}(d)$ at any separation $d$.
\begin{theorem}
\label{thm:coordinate-kernel-exact}
Let $f\colon V\to\F_q^\ell$ be a surjective linear function. Assume that $\ker f = \operatorname{span}\{e_i:i\in I\}$ for some
$I\subseteq\{1,\dots,k\}$ with $|I|=k-\ell$.
Then
\[ 
 r_f^{\mathrm{sys,lin}}(d)
 =
 r_f^{\mathrm{free,lin}}(d)
 =
 \Nlin(\ell,d)-\ell. 
\]

\end{theorem}

\begin{proof}
For convenience, let $N=\Nlin(\ell,d)$. By Theorem~\ref{thm:free-exact} and
Corollary~\ref{cor:systematic-lower},
\[
r_f^{\mathrm{sys,lin}}(d)\geq r_f^{\mathrm{free,lin}}(d)=N-\ell.
\]
Let $J=I^c$. 
By Theorem \ref{thm:monomial-invariance}, we may replace $f$ by a monomially equivalent function and assume
$I=\{1,\dots,k-\ell\}$ and $J=\{k-\ell+1,\dots,k\}$. After a permutation of transmitted coordinates, a
linear $[N,\ell,\geq d]_q$ code has a systematic generator matrix
$[\,I_\ell\mid Q\,]$ with $Q\in\F_q^{\ell\times(N-\ell)}$. Define
the linear encoding $\Enc\colon V\to\F_q^{k+N-\ell}$ by the
generator matrix
\[
G_{\mathcal E}(\Enc)
=
\left[\,I_k\;\middle|\;
\begin{matrix}0\\ Q\end{matrix}\,\right].
\]
For $w\in V$, write $[w]_{\mathcal E}=(x,z)$ with
$x\in\F_q^{k-\ell}$ and $z\in\F_q^{\ell}$, so that
$\Enc(w)=(x,z,zQ)$. If $\pi_J^{\mathcal E}(w)=z\neq0$, then
$(z,zQ)$ is a nonzero codeword of the $[N,\ell,\geq d]_q$ code
generated by $[\,I_\ell\mid Q\,]$, so
\[
\wt\bigl(\Enc(w)\bigr)
=
\wt(x)+\wt(z,zQ)
\geq d.
\]
Hence $\delta_{\Enc}(\pi_J^{\mathcal E})\geq d$. Since $f$ and
$\pi_J^{\mathcal E}$ are surjective with the same kernel,
Corollary~\ref{cor:kernel-invariance} gives
$\delta_{\Enc}(f)\geq d$, and therefore
$r_f^{\mathrm{sys,lin}}(d)\leq N-\ell$.
\end{proof}

\section{The linear separation problem at small separations}
\label{sec:distance-three}

We now determine the free and systematic linear redundancies for the first three prescribed separations $d=1,2,3$. For $d=1$, the answer is immediate: every encoding
has $\delta_{\Enc}(f)\geq1$, and so $r_f^{\mathrm{sys,lin}}(1)=r_f^{\mathrm{free,lin}}(1)=0$.
For $d=2$, we have the following. 
\begin{proposition}
\label{prop:small-separations}
Let $f\colon V\to\F_q^\ell$ be a surjective linear function. Then
\[
r_f^{\mathrm{sys,lin}}(2)=r_f^{\mathrm{free,lin}}(2)=1.
\]
\end{proposition}

\begin{proof}
By the Singleton bound, we have $\Nlin(\ell,2)\ge\ell+1,$
and by Theorem  \ref{thm:free-exact}, it follows that $r_f^{\mathrm{free,lin}}(2)\geq1$. Conversely,
define the parity map $p\colon V\to\F_q$ by
\[
p(v):=\sum_{\substack{i\\ f(e_i)\neq0}}\bigl([v]_{\mathcal E}\bigr)_i,
\]
and let $\Enc\colon V\to\F_q^{k+1}$ be the linear encoding,
systematic relative to $\mathcal E$, given by
$\Enc(v):=\bigl([v]_{\mathcal E},p(v)\bigr)$. If $f(w)\neq0$ and
$\wt([w]_{\mathcal E})\geq2$, then $\wt(\Enc(w))\geq2$. If
$\wt([w]_{\mathcal E})=1$, then $w=ce_i$ with $c\neq0$ and
$f(e_i)\neq0$, so $p(w)=c\neq0$ and $\wt(\Enc(w))=2$. Hence
$\delta_{\Enc}(f)\geq2$, and
$1\leq r_f^{\mathrm{free,lin}}(2)\leq
r_f^{\mathrm{sys,lin}}(2)\leq1$.
\end{proof}

Thus the systematicity gap is zero  at prescribed separations
$d=1$ and $d=2$. In the remainder of this section we will settle the next case $d=3$.

For a nonzero vector $x$ in a finite-dimensional $\F_q$-vector
space, let $\langle x\rangle$ denote the projective point spanned by
$x$. In $\mathrm{PG}(r-1,q)$, we call the points spanned by the
standard basis vectors of $\F_q^r$ the \emph{coordinate points},
and every other point a \emph{non-coordinate point}. There are
\[
\Lambda_q(r) := \frac{q^r-1}{q-1}-r
\]
non-coordinate points. 
Equivalently, $\Lambda_q(r)$ is the number of
points of $\mathrm{PG}(r-1,q)$ whose spanning vectors have
Hamming weight at least two.

For a linear function $f\colon V\to\F_q^\ell$, let
$f_i:=f(e_i)\in\F_q^\ell$ for $1\leq i\leq k$, and let
\[
\operatorname{supp}(f):=\{i\in\{1,\dots,k\}: f_i\neq0\}.
\]

\begin{definition}
\label{def:projective-coordinate-number}
Let $f\colon V\to\F_q^\ell$ be a surjective linear function. Define
\[
\Phi_{\mathcal E}(f)
:=
\{\langle f_i\rangle : i\in\operatorname{supp}(f)\}
\subseteq\mathrm{PG}(\ell-1,q),
\qquad
\nu_{\mathcal E}(f):=|\Phi_{\mathcal E}(f)|.
\]
\end{definition}

\begin{remark}
\label{rem:projective-coordinate-invariance}
We note that $ \operatorname{supp}(f)=\{i:e_i\notin\ker f\}$,
and  for $i,j\in\operatorname{supp}(f)$,
$$
\langle f_i\rangle=\langle f_j\rangle
\quad\Longleftrightarrow\quad
e_i-\lambda e_j\in\ker f
\quad\text{for some }\lambda\in\F_q^\times.
$$
Consequently, $\nu_{\mathcal E}(f)$ is determined by $\ker f$.
Moreover, if $f$ and $g$ are monomially equivalent, then the rows of
$M_g$ are obtained from those of $M_f$ by permutation and nonzero
scaling, and hence
$\nu_{\mathcal E}(g)=\nu_{\mathcal E}(f).$
We have the  relations
\[
\ker f=\ker g
\quad\Longrightarrow\quad
\nu_{\mathcal E}(f)=\nu_{\mathcal E}(g),
\]
\[
\text{$f$ and $g$ are monomially equivalent}
\quad\Longrightarrow\quad
\nu_{\mathcal E}(f)=\nu_{\mathcal E}(g),
\]
but neither converse is true. For example, over $\F_2$, the functions $f$ and $g$ represented by
\[
M_f=
\begin{pmatrix}
1&0\\1&0\\1&0\\0&1
\end{pmatrix},
\qquad
M_g=
\begin{pmatrix}
1&0\\1&0\\0&1\\0&1
\end{pmatrix}
\]
 have $\nu_{\mathcal E}(f)=\nu_{\mathcal E}(g)=2$, but have different kernels and are not monomially equivalent.
\end{remark}

\begin{proposition}
\label{prop:free-distance-three}
Let $f\colon V\to\F_q^\ell$ be a surjective linear function. Then
\[
r_f^{\mathrm{free,lin}}(3)
=
\min\{r\geq1:\ell\leq\Lambda_q(r)\}.
\]
\end{proposition}

\begin{proof}
By Theorem~\ref{thm:free-exact}, it is sufficient to determine when a linear $[r+\ell,\ell,\geq3]_q$ code $C$ exists. If such $C$ exists, then $C^\perp$ is a projective code where the columns of its generator matrix form a set of distinct points in $\mathrm{PG}(r-1,q)$ (see for example \cite[p.84]{ding2018}). 
This means
$r+\ell\leq\frac{q^r-1}{q-1}$, and so $\ell\leq\Lambda_q(r)$.

Conversely, assume that $\ell\leq\Lambda_q(r)$.
Consider the Hamming code $\mathcal H_{q,r}$ (see~\cite[p.81]{ding2018} or~\cite[p. 29]{huffman2003}), whose
parity-check matrix has as columns representatives of all $\frac{q^r-1}{q-1}=r+\Lambda_q(r)$ points of $\mathrm{PG}(r-1,q)$.
Let $H'$ be the $r\times(r+\ell)$ submatrix whose columns represent
the $r$ coordinate points together with any $\ell$ non-coordinate
points, and let $C'$ be the code with parity-check matrix $H'$.
Then $C'$ is a linear $[r+\ell,\ell,\geq3]_q$ code.
\end{proof}

\begin{lemma}
\label{lem:distance-three-rows}
Let $\Enc\colon V\to C\subseteq\F_q^{k+r}$ be a linear encoding,
systematic relative to $\mathcal E$, with
$G_{\mathcal E}(\Enc)=[\,I_k\mid P\,]$, and let
$p_1,\dots,p_k\in\F_q^r$ denote the rows of $P$. Let
$f\colon V\to\F_q^\ell$ be a linear function. If
$\delta_{\Enc}(f)\geq3$, then for all
$i,j\in\operatorname{supp}(f)$:
\begin{enumerate}
\item $\wt(p_i)\geq2$. In particular, $\langle p_i\rangle$ is a
non-coordinate point of $\mathrm{PG}(r-1,q)$.
\item If $\langle f_i\rangle\neq\langle f_j\rangle$, then
$\langle p_i\rangle\neq\langle p_j\rangle$.
\end{enumerate}
\end{lemma}

\begin{proof}
Since $\delta_{\Enc}(f)\geq3$, whenever $f(w)\neq0$, we have
\[
\wt\bigl([w]_{\mathcal E}\bigr)+\wt\bigl([w]_{\mathcal E}P\bigr)
\geq3.
\]
For $w=e_i$, we have $[w]_{\mathcal E}P=p_i$ and
$f(w)=f_i\neq0$, so $1+\wt(p_i)\geq3$, which implies $\wt(p_i)\geq2$.

For  $w=\alpha e_i+\beta e_j$ with $\alpha,\beta\neq0$, we have
$[w]_{\mathcal E}P=\alpha p_i+\beta p_j$. If $\langle f_i\rangle\neq\langle f_j\rangle$, then $f_i$ and $f_j$ are linearly
independent, and  $f(w)=\alpha f_i+\beta f_j\neq0$. 
Then 
\[
\wt\bigl([w]_{\mathcal E}P\bigr) \ge 
3-\wt\bigl([w]_{\mathcal E}\bigr) = 
3-2=1, 
\]
so that $\alpha p_i+\beta p_j\neq0$. Hence the rows $p_i$
and $p_j$ are linearly independent and
$\langle p_i\rangle\neq\langle p_j\rangle$.
\end{proof}

\begin{theorem}
\label{thm:systematic-distance-three}
Let $f\colon V\to\F_q^\ell$ be a surjective linear function. For
every $r\geq1$, the following are equivalent:
\begin{enumerate}
\item There exists a linear encoding
$\Enc\colon V\to C\subseteq\F_q^{k+r}$, systematic relative to
$\mathcal E$, with $\delta_{\Enc}(f)\geq3$.
\item $\nu_{\mathcal E}(f)\leq\Lambda_q(r)$.
\end{enumerate}
Consequently,
\[
r_{f,\mathcal E}^{\mathrm{sys,lin}}(3)
=
\min\{r\geq1:\nu_{\mathcal E}(f)\leq\Lambda_q(r)\}.
\]
\end{theorem}

\begin{proof}
Let $\Enc$ be as in~(1), with
$G_{\mathcal E}(\Enc)=[\,I_k\mid P\,]$, and let
$p_1,\dots,p_k$ denote the rows of $P$. For each
$L\in\Phi_{\mathcal E}(f)$, fix an index
$i_L\in\operatorname{supp}(f)$ with $\langle f_{i_L}\rangle=L$. 
Define
\[
\psi\colon\Phi_{\mathcal E}(f)\longrightarrow\mathrm{PG}(r-1,q),
\qquad
\psi(L):=\langle p_{i_L}\rangle.
\]
By Lemma~\ref{lem:distance-three-rows}, $\psi$ is an injective map into the set of non-coordinate points of $\mathrm{PG}(r-1,q)$. Hence
$\nu_{\mathcal E}(f)\leq\Lambda_q(r)$.

Conversely, assume that 
$\nu_{\mathcal E}(f)\leq\Lambda_q(r)$, and choose an injection
\[
\gamma\colon\Phi_{\mathcal E}(f)\longrightarrow
\{\text{non-coordinate points of }\mathrm{PG}(r-1,q)\}.
\]

For each point $L\in\Phi_{\mathcal E}(f)$, we fix a vector $h_L\in\F_q^r$
spanning $\gamma(L)$ and an index $i_L\in\operatorname{supp}(f)$
with $\langle f_{i_L}\rangle=L$.

Let $P\in\F_q^{k\times r}$ be the matrix with rows
$p_i, 1 \le i\le k$, defined as follows.
For $i\notin\operatorname{supp}(f)$, set
$p_i:=0$. 
For each $i\in\operatorname{supp}(f)$, let $L:=\langle f_i\rangle$,
and   write $f_i=\lambda_i f_{i_L}$ with $\lambda_i\in\F_q^\times$.
Set $p_i:=\lambda_ih_L$. 

Let $\Enc$ be the linear encoding with
$G_{\mathcal E}(\Enc)=[\,I_k\mid P\,]$.
This encoding is systematic relative to $\mathcal E$ and has redundancy $r$.
We now verify $\delta_{\Enc}(f)\geq3$ by checking
$\wt([w]_{\mathcal E})+\wt([w]_{\mathcal E}P)\geq3$ for all $w$
with $f(w)\neq0$. It is sufficient to check only those $w$ with $\wt([w]_{\mathcal E}) \le 2$. 

\begin{enumerate}[wide=0pt, leftmargin=*, label=Case~\arabic*:]

\item $w=ce_i$ with $c\neq0$. Here $f(w)=cf_i\neq0$ implies
$i\in\operatorname{supp}(f)$; put $L:=\langle f_i\rangle$. Then
$[w]_{\mathcal E}P=c\lambda_ih_L$ has weight $\wt(h_L)\geq2$, since
$\gamma(L)$ is non-coordinate. So
$\wt(\Enc(w))=\wt([w]_{\mathcal E})+\wt([w]_{\mathcal E}P)\geq3$.

\item  $w=\alpha e_i+\beta e_j$ with $\alpha,\beta\neq0$ and
$f(w)\neq0$, say $i\in\operatorname{supp}(f)$, with
$L=\langle f_i\rangle$. 

If $j\notin\operatorname{supp}(f)$, then
$[w]_{\mathcal E}P=\alpha\lambda_ih_L\neq0$.

If
$j\in\operatorname{supp}(f)$ with $\langle f_j\rangle=L$, then
$f(w)=(\alpha\lambda_i+\beta\lambda_j)f_{i_L}\neq0$, which implies
$\alpha\lambda_i+\beta\lambda_j\neq0$, and
$[w]_{\mathcal E}P=(\alpha\lambda_i+\beta\lambda_j)h_L\neq0$. 

If
$j\in\operatorname{supp}(f)$ with
$\langle f_j\rangle=L'\neq L$, then $h_L$ and $h_{L'}$ span the
distinct points $\gamma(L)\neq\gamma(L')$, so
$[w]_{\mathcal E}P=\alpha\lambda_ih_L+\beta\lambda_j h_{L'}\neq0$.
\end{enumerate} 
In each subcase of Case 2, we have $\wt([w]_{\mathcal E}P) \ge 1$, and so $\wt(\Enc(w)) \ge 3$. Hence $\delta_{\Enc}(f)\geq3$.
\end{proof}

\begin{remark}
Corollary~\ref{cor:kernel-invariance} shows that systematic linear
redundancy is unchanged when $f$ is replaced by a linear function
with the same kernel.
Theorem~\ref{thm:monomial-invariance} shows that it is unchanged
under monomial equivalence of functions. 
At separation $d=3$,
Theorem~\ref{thm:systematic-distance-three} gives a stronger description with the invariant $\nu_{\mathcal E}(f)$, see Remark \ref{rem:projective-coordinate-invariance}.
\end{remark}

\begin{remark}
\label{rem:clark-small-distance} 
In the UEP literature, 
Clark, Dunning, and Rogers~\cite{clark1990} studied a small-distance problem different from ours. 
As an application of their results on binary parity-check matrices,
they determined, for each fixed redundancy $r$, exact block-length
thresholds for the existence of binary linear codes of minimum
distance three or four with a nonconstant input separation vector.
In contrast, we fix a linear function $f$ and a prescribed
separation $d$, and ask for the optimal redundancy.
\end{remark}

 From Proposition~\ref{prop:free-distance-three} and Theorem~\ref{thm:systematic-distance-three}, we have the following characterization when the systematicity gap is zero at separation $d=3$.
\begin{corollary}
\label{cor:systematic-equality-distance-three}
Let $f\colon V\to\F_q^\ell$ be a surjective linear function. Let 
\[
r_0
:=
r_f^{\mathrm{free,lin}}(3)
=
\Nlin(\ell,3)-\ell
=
\min\{r\geq1:\ell\leq\Lambda_q(r)\}.
\]
Then
\[ 
r_{f,\mathcal E}^{\mathrm{sys,lin}}(3)
=
r_f^{\mathrm{free,lin}}(3)
\quad\Longleftrightarrow\quad
\nu_{\mathcal E}(f)\leq\Lambda_q(r_0).
\]
\end{corollary}

We illustrate with two examples.

\begin{example}
\label{ex:noncoordinate-equality}
Let $V=\F_2^4$ with its standard basis as $\mathcal E$, and define
\[
f\colon V\longrightarrow\F_2^3,
\qquad
f(x_1,x_2,x_3,x_4)
=
(x_1+x_4,\,x_2+x_4,\,x_3+x_4).
\]
The vectors $f_i=f(e_i)$ are
\[
f_1=(1,0,0),\quad f_2=(0,1,0),\quad f_3=(0,0,1),\quad f_4=(1,1,1),
\]
so $\operatorname{supp}(f)=\{1,2,3,4\}$.
The set
$\Phi_{\mathcal E}(f)$ consists of the four points
$\langle f_1\rangle,\dots,\langle f_4\rangle$ and
$\nu_{\mathcal E}(f)=4$. Since $\Lambda_2(2)=1$ and
$\Lambda_2(3)=4$, Theorem~\ref{thm:systematic-distance-three} gives
$r_f^{\mathrm{sys,lin}}(3)=3$, and
Proposition~\ref{prop:free-distance-three} gives
$r_f^{\mathrm{free,lin}}(3)=3$.

We also describe the  encoding $\Enc$ of
Theorem~\ref{thm:systematic-distance-three} in this case. 
The non-coordinate points of $\mathrm{PG}(2,2)$ are
\[
\langle(1,1,0)\rangle,\quad
\langle(1,0,1)\rangle,\quad
\langle(0,1,1)\rangle,\quad
\langle(1,1,1)\rangle.
\]
Since
$\nu_{\mathcal E}(f)=4$, the injection $\gamma$ is
forced to be a bijection. We define $\gamma$ as
\[
\gamma(\langle f_1\rangle)=\langle(1,1,0)\rangle,\quad
\gamma(\langle f_2\rangle)=\langle(1,0,1)\rangle,\quad
\gamma(\langle f_3\rangle)=\langle(0,1,1)\rangle,\quad
\gamma(\langle f_4\rangle)=\langle(1,1,1)\rangle.
\]
Over $\F_2$, we have $\lambda_i=1$ for all $i$, and  each point $\gamma(L)$ has a unique representative $h_L$. 
The $i$th row of $P$ is then
$p_i=h_{\langle f_i\rangle}$, so that
\[
P=
\begin{pmatrix}
1&1&0\\
1&0&1\\
0&1&1\\
1&1&1
\end{pmatrix},
\]
and the encoding with $G_{\mathcal E}(\Enc)=[\,I_4\mid P\,]$ has
$\delta_{\Enc}(f)\geq3$.

In this example,  $r_f^{\mathrm{sys,lin}}(3)=r_f^{\mathrm{free,lin}}(3)=3$ and the systematicity gap is zero. 
On the other hand, we see that $\ker f=\operatorname{span}\{(1,1,1,1)\}$ is not spanned by basis vectors of $\mathcal{E}$. In particular, the sufficient condition of
Theorem~\ref{thm:coordinate-kernel-exact} is not necessary for
$r_f^{\mathrm{sys,lin}}(3)=r_f^{\mathrm{free,lin}}(3)$.
\end{example} 

\begin{example}
\label{ex:strict-systematicity-gap}
Let $V=\F_3^4$ with its standard basis as $\mathcal E$, and define
\[
f\colon V\longrightarrow\F_3^2,
\qquad
f(x_1,x_2,x_3,x_4)
=
(x_1+x_3+2x_4,\,x_2+x_3+2x_4).
\]
The vectors $f_i=f(e_i)$ are
\[
f_1=(1,0),\qquad
f_2=(0,1),\qquad
f_3=(1,1),\qquad
f_4=(2,2)=2f_3,
\]
so $\operatorname{supp}(f)=\{1,2,3,4\}$. Then
\[
\Phi_{\mathcal E}(f)
=
\bigl\{
\langle(1,0)\rangle,\,
\langle(0,1)\rangle,\,
\langle(1,1)\rangle
\bigr\},
\qquad
\nu_{\mathcal E}(f)=3.
\]
Since $ \Lambda_3(2)=2,   \Lambda_3(3)=10, $
Proposition~\ref{prop:free-distance-three} and
Theorem~\ref{thm:systematic-distance-three} give
\[
r_f^{\mathrm{free,lin}}(3)=2
<
3=r_f^{\mathrm{sys,lin}}(3).
\]
There are 10 non-coordinate points in $\mathrm{PG}(2,3)$. We define $\gamma$ as
\[
\begin{aligned}
\gamma(\langle(1,0)\rangle)&=\langle(1,1,0)\rangle,\\
\gamma(\langle(0,1)\rangle)&=\langle(1,0,1)\rangle,\\
\gamma(\langle(1,1)\rangle)&=\langle(0,1,1)\rangle.
\end{aligned}
\]
For the points $\langle(1,0)\rangle$, $\langle(0,1)\rangle$, and
$\langle(1,1)\rangle$ of $\Phi_{\mathcal E}(f)$, we fix the
representatives $h_L=(1,1,0)$, $(1,0,1)$, and $(0,1,1)$ of their
images under $\gamma$, and the indices $i_L=1$, $2$, and $3$,
respectively. The scalars of Theorem~\ref{thm:systematic-distance-three}
are then $\lambda_1=\lambda_2=\lambda_3=1$ and $\lambda_4=2$, since
$f_4=2f_3$.
The construction in
Theorem~\ref{thm:systematic-distance-three} gives
\[
P=
\begin{pmatrix}
1&1&0\\
1&0&1\\
0&1&1\\
0&2&2
\end{pmatrix}.
\]
Then the
encoding $\Enc$ with $G_{\mathcal E}(\Enc)=[\,I_4\mid P\,]$ has $\delta_{\Enc}(f)\geq3$. 
\end{example}

Examples~\ref{ex:noncoordinate-equality}
and~\ref{ex:strict-systematicity-gap} exhibit systematicity gaps
$r_f^{\mathrm{sys,lin}}(3)-r_f^{\mathrm{free,lin}}(3)$ of zero and one. 
For fixed $q$, this gap can be arbitrarily large.
\begin{corollary} 
\label{cor:unbounded-systematicity-gap}
Let $q$ be a prime power and $\ell\geq1$, and let $V$ be an $\F_q$-vector space of dimension $k=\frac{q^\ell-1}{q-1}$ with a
distinguished basis $\mathcal E$. Then there exists a surjective
linear function $f\colon V\to\F_q^\ell$ with systematicity gap
\[
r_f^{\mathrm{sys,lin}}(3)-r_f^{\mathrm{free,lin}}(3)
\geq
\ell-1-\lceil\log_q\ell\rceil.
\]
In particular, for every $N>0$, there exist $\ell \ge1$  and $V, \mathcal{E},f$ as above such that
\[
r_f^{\mathrm{sys,lin}}(3)-r_f^{\mathrm{free,lin}}(3)>N.
\]
\end{corollary}

\begin{proof}
Choose nonzero vectors $f_1,\dots,f_k\in\F_q^\ell$, one
representative from each point of $\mathrm{PG}(\ell-1,q)$, indexed
so that $f_1,\dots,f_\ell$ are the standard basis vectors of
$\F_q^\ell$. Let $f\colon V\to\F_q^\ell$ be the unique surjective linear
function satisfying
\[
f(e_i)=f_i,\qquad 1\leq i\leq k.
\]
Moreover, since we choose $f_i$ nonzeros, 
we have 
$\operatorname{supp}(f)=\{1,\dots,k\}.$
Also, $\Phi_{\mathcal E}(f)=\mathrm{PG}(\ell-1,q),$
and so $\nu_{\mathcal E}(f)=k.$
For $r\leq\ell$, we have
\[
\Lambda_q(r)<\frac{q^r-1}{q-1}\leq k,
\]
while $\Lambda_q(\ell+1)=q^\ell+k-(\ell+1)\geq k.$
By Theorem~\ref{thm:systematic-distance-three},  $r_f^{\mathrm{sys,lin}}(3)=\ell+1$.

Let $m:=\lceil\log_q\ell\rceil$, so that $q^m\geq\ell$. Since
\[
\Lambda_q(m+2) =
q^{m+1}+q^m+\cdots+q+1-(m+2)
\geq q^{m+1}-1 \geq q\ell-1 \geq \ell,
\]
Proposition~\ref{prop:free-distance-three} gives
$r_f^{\mathrm{free,lin}}(3)\leq m+2$.
Hence,
\[
r_f^{\mathrm{sys,lin}}(3)-r_f^{\mathrm{free,lin}}(3)
\geq
\ell+1-(m+2)= \ell-1-\lceil\log_q\ell\rceil,
\]
which is the stated bound. The last statement follows from classical analysis.
\end{proof}
\begin{remark}
\label{rem:distance-four}
In Theorem~\ref{thm:free-exact}, we have determined the free linear
redundancy for every prescribed separation. The  systematic linear separation problem at $d=4$ is more involved. Indeed, if
$G_{\mathcal E}(\Enc)=[\,I_k\mid P\,]$, then every $w\in V$ with
$f(w)\neq0$ and
$\wt([w]_{\mathcal E})=s\leq3$ must satisfy
\[
\wt\bigl([w]_{\mathcal E}P\bigr)\geq4-s.
\]
Thus one must control the Hamming weights of all relevant linear
combinations of up to three rows of $P$, rather than only the
projective conditions used at separation $d=3$. We leave an explicit
determination of
$r_{f,\mathcal E}^{\mathrm{sys,lin}}(4)$ for future work.
\end{remark} 

\section{Conclusion}\label{sec:conclusion}

In this paper, we developed a function-separation framework for FCCs and UEP codes. Our main contributions are the following.

\begin{enumerate}

\item The relation between FCCs and UEP is known from
\cite{lenz2023}. We made it precise at the level of parameters: the
output-coordinate, input-coordinate, and component UEP parameters are
function-separation distances of the corresponding coordinate and
projection maps.
For linear encodings and
linear functions, we identified $\delta_{\Enc}(f)$ with the first relative
generalized Hamming weight of the pair
$\Enc(\ker f)\subsetneq C$, and obtained its separation-vector
representation with respect to a suitable basis
(Theorems~\ref{thm:kernel-formula} and~\ref{thm:adapted-basis}).

\item We formulated the free and systematic linear separation problems and introduced the optimal
free and systematic linear redundancies
$r_f^{\mathrm{free,lin}}(d)$ and
$r_f^{\mathrm{sys,lin}}(d)$. We proved that the optimal free linear
redundancy is the same for all linear functions of the same rank
(Theorem~\ref{thm:free-exact}), and that monomially equivalent functions
have the same optimal systematic linear redundancy
(Theorem~\ref{thm:monomial-invariance}).

\item We determined the optimal systematic linear redundancy for all
prescribed separations $d\leq3$. The cases $d=1$ and $d=2$ are
determined in Proposition~\ref{prop:small-separations}, while the
case $d=3$ is determined in
Theorem~\ref{thm:systematic-distance-three}. Combining these results
with the free linear redundancy formula of
Theorem~\ref{thm:free-exact}, we characterized exactly when the
systematicity gap at separation $d=3$ is zero and showed that this
gap can be arbitrarily large
(Corollaries~\ref{cor:systematic-equality-distance-three}
and~\ref{cor:unbounded-systematicity-gap}).

\end{enumerate}

These results show that systematicity is a genuine
structural constraint in linear function-correcting codes. 
The systematicity gap depends on the position of $\ker f$ relative to the distinguished basis  and can already be arbitrarily large at separation $d=3$. 
Determining the systematicity gap at separation $d \ge 4$ is a natural next step in the development of the function-separation framework.

\end{document}